\documentclass[acmsmall,screen,nonacm]{acmart}

\usepackage[noend]{algpseudocode}
\usepackage{algorithm}
\usepackage{amsmath,amssymb,amsfonts}
\usepackage{balance}
\usepackage{cleveref}
\usepackage{diagbox}
\usepackage{graphicx}
\usepackage{multirow}
\usepackage{physics}
\usepackage{rotating,tabularx}
\usepackage{soul}
\usepackage{subfig}
\usepackage{textcomp}
\usepackage{threeparttable}
\usepackage{tikz}
\usepackage{wrapfig}
\usepackage{xcolor}
\usepackage{xfrac}
\usepackage{xspace}

\def\BibTeX{{\rm B\kern-.05em{\sc i\kern-.025em b}\kern-.08em
    T\kern-.1667em\lower.7ex\hbox{E}\kern-.125emX}}

\newcommand{\sys}{QBX\xspace}
\newcommand{\tket}{t$|$ket$\rangle$}
\theoremstyle{definition}
\newtheorem{definition}{Definition}[section]
\newtheorem{theorem}{Theorem}

\title{\sys: A Compiler for 2-local Qubit Hamiltonian Simulation on Quantum Chiplets}
\author{Zikun Li}
\email{zikunl@andrew.cmu.edu}
\affiliation{\institution{Carnegie Mellon University}
  \city{Pittsburgh}
  \country{USA}
}
\author{Zhuoming Chen}
\email{zhuominc@andrew.cmu.edu}
\affiliation{\institution{Carnegie Mellon University}
  \city{Pittsburgh}
  \country{USA}
}
\author{Zhihao Jia}
\email{zhihao@cmu.edu}
\affiliation{\institution{Carnegie Mellon University}
  \city{Pittsburgh}
  \country{USA}
}

\begin{document}

\begin{abstract}

2-local qubit Hamiltonian simulation, a fundamental task in quantum computing, is widely applied in various applications. This paper presents \sys, the first quantum compiler designed for 2-local qubit Hamiltonian simulation on quantum chiplet architectures.
Existing general-purpose quantum compilers for chiplet architectures target quantum programs at the gate level and miss optimizations requiring high-level semantics of Hamiltonian simulation.
Conversely, existing domain-specific compilers for Hamiltonian simulation are designed for monolithic architectures, and their limited scalability and inability to adapt to the heterogeneity of chiplet architectures lead to sub-optimal circuit outputs.
\sys implements a scalable hierarchical approach to reduce both the amount and distance of cross-chiplet communications. It integrates the {\em highway} mechanism, a chiplet-oriented solution for efficient long-range qubit communication, to diminish cross-chiplet interaction costs. Furthermore, \sys efficiently aggregates 2-local Pauli strings into multi-target controlled gates, harnessing highways for deep optimization.
Our evaluations show that \sys markedly outperforms both general-purpose and domain-specific compilers in circuit depth and operation count, while scaling to much larger 2-local Hamiltonian simulations.

\end{abstract}
\maketitle

\section{Introduction}

Quantum computing provides great potential to achieve significant acceleration over classical computers, representing a computational paradigm shift in various domains, as highlighted in seminal works~\cite{grover1996fast, shor1999polynomial, qft, hhl}.
However, existing quantum devices, particularly those in the noisy intermediate-scale quantum (NISQ) era, face significant challenges, including  low qubit capacity, non-negligible operation noise, and qubit decoherence~\cite{preskill2018quantum}.
These limitations make executing complex quantum algorithms, such as Shor's~\cite{shor1999polynomial} and HHL~\cite{hhl}, impractical in the near term.

In the realm of practical algorithms suited for NISQ devices, a standout candidate is {\em Hamiltonian simulation}, which aims to simulate the time evolution of a physical system relative to its Hamiltonian~\cite{feynman2018simulating}.
This simulation is a daunting task for classical computers due to exponential computational complexity.
In particular, the 2-local qubit form of Hamiltonian simulation is widely used in many applications, such as
the Ising model for studying ferromagnetism in materials \cite{cipra1987introduction, newell1953theory}, the XY model in phase transition studies related to superfluidity and superconductivity \cite{babaev1999nonperturbative, ohta1979xy}, and the Heisenberg model for exploring quantum magnetic properties \cite{sachdev1993gapless}. In addition, the quantum approximate optimization algorithm (QAOA) \cite{qaoa} also utilizes a 2-local qubit Hamiltonian structure.

Supporting Hamiltonian simulations on today's quantum devices is challenging.
Though they require a modest number of qubits and qubit coherence time, only minimal problem instances are currently feasible to run on quantum devices.
{\em Superconducting-based} quantum computing is a leading technology pathway, offering up to 433 qubits on a single chip~\cite{ibm}.
However, scaling quantum chips with a monolithic architecture faces obstacles such as increased frequency collision likelihood and reduced manufacturing yield \cite{smith2022scaling}.
Additionally, the capacity limits of cryogenic dilution refrigerators impose constraints on chip scale \cite{krinner2019engineering, ang2022architectures}.
These challenges have spurred interest in {\em chiplet architectures}, where multiple smaller, interconnected chiplets form a multi-chip module (MCM) \cite{smith2022scaling, chow2021quantum, laracuente2022modeling}.
This architecture reduces frequency collision rates and enhances manufacturing yield \cite{smith2022scaling}.
Recent advancements in short-range, inter-chip connections \cite{wallraff2018deterministic, zhong2021deterministic, zhou2021modular, magnard2020microwave} also promote this modular approach.
Consequently, the chiplet architecture is gaining significant attention in both research and industry circles as a viable solution for scaling quantum capacity in the near term.

Compiling quantum programs for chiplet architectures introduces additional challenges and complexities.
Multi-chip modules, in contrast to monolithic chips, significantly expand the device's scale, necessitating a reassessment of communication costs.
Existing quantum compilers, such as Qiskit \cite{qiskit} and \tket \cite{tket}, realize inter-qubit communication through qubit routing, which will largely increase the depth and number of SWAPs in the compiled circuits as the quantum device size grows.
Another challenge arises from the heterogeneity of cross-chiplet couplings, which are significantly noisier than intra-chiplet couplings.
Quantum compilers for monolithic chips generally ignore this heterogeneity, resulting in circuits with suboptimal performance and fidelity.
A notable advancement is {\em highway}~\cite{highway}, a chiplet-specific mechanism for efficient long-range communication between qubits (see \Cref{subsec:highway} for an introduction).
Using highway for cross-chiplet communication minimizes the cost of qubit routing and avoids intensive usage of cross-chiplet couplings.
Highway also supports multi-target controlled gates where multiple target qubits are simultaneously entangled with the same control qubit, increasing program concurrency.

Existing quantum compilers for chiplet architectures (e.g., MECH~\cite{highway}) are designed to optimize general quantum programs at the gate level, and therefore miss domain-specific optimizations that require high-level semantics of Hamiltonian simulation, such as commutativity of Pauli strings. On the other hand, today's domain-specific compilers for Hamiltonian simulation~\cite{2qan, paulihedral, dallaire2016quantum, shi2019optimized} target monolithic architectures, and generate suboptimal circuits for chiplet architectures.

To bridge this gap, this paper introduces \sys, the first domain-specific quantum compiler for 2-local qubit Hamiltonian simulation on chiplet architectures.
Different from MECH\cite{highway}, which dynamically schedules the building and consumption of highways, \sys adopts a {\em static} approach to leveraging highways in a more structured way.
Specifically, since highways are expensive to construct,
\sys uses highways for (1) cross-chiplet communications, and (2) intra-chiplet communications where the benefit of using a highway outweighs its cost.
To minimize cross-chiplet communications, \sys introduces a scalable {\em hierarchical mapping} approach that decomposes qubit mapping into two tasks: one assigns qubits to chiplets, and the other maps all qubits of a chiplet to its registers.
The rationale behind this approach is that the problem of assigning qubits to chiplets involves a much smaller scale than the original mapping problem, allowing \sys to directly minimize cross-chiplet communications and their distances.
In addition, based on a key observation that multiple 2-local Pauli strings can be transformed into multi-target controlled gates, which can be executed together using highways, \sys's {\em Pauli block scheduler} leverages the commutativity of Pauli strings to systematically aggregate 2-local Pauli strings and schedules the aggregated circuits toward minimal depth.

The aforementioned techniques allows \sys to significantly outperform existing quantum compilers for 2-local qubit Hamiltonian simulation.
Specifically, \sys reduces the depth of compiled circuits by up to 8.9$\times$ and 53.2$\times$ compared with Qiskit \cite{qiskit} and \tket \cite{tket}, two widely used general-purpose quantum compilers.
When compared with 2QAN~\cite{2qan}, a domain-specific compiler for 2-local qubit Hamiltonian simulation, \sys achieves on par performance and is 19.5$\times$ faster than 2QAN for small problems, while scaling to significantly larger problems.

\section{Background and Related Works}
\label{sec:background}

\subsection{Background}

\paragraph{Pauli strings}
For an $n$-qubit system, a Pauli string is a tensor product of $n$ Pauli operators (i.e. the $X$, $Y$ and $Z$ operators) or the identity operator, that is, $P = \otimes_{i=0}^{n-1} \sigma_i$, where $\sigma_i \in \{I, X, Y, Z\}$.
The $i$-th operator corresponding to the $i$-th qubit.
In this paper, we denote a Pauli string in the format of $\sigma_{n-1}\dots\sigma_1\sigma_0$ (e.g. $I_2Y_1Z_0)$, where subscript denotes the index of qubit corresponding to the operator, and some times identity are omitted for simplicity (e.g. $Z_2Z_0$).
We also use the {\em type} of a 2-local Pauli string to refer to its non-identity operator type (e.g. $Z_2Z_0$ has type $ZZ$).
A special type of Pauli string is \textit{$k$-local} Pauli strings, which apply non-identity Pauli operators to exactly $k$ qubits.
For instance, $Y_1X_0$ and $Y_2Z_0$ are 2-local, while $Y_3X_2X_1$ is not.

\paragraph{Hamiltonian Simulation}

The development of quantum computers is significantly motivated by Hamiltonian simulation \cite{feynman2018simulating}.This process involves simulating the state of a system, denoted as $\ket{\psi(t)}$, at a specific time $t$ from its initial state $\ket{\psi(0)}$.
According to the Schrödinger equation, $\ket{\psi(t)} = e^{-iHt}\ket{\psi(0)}$, where $H$ is the system's Hamiltonian.
The primary challenge is directly devising a quantum circuit that synthesizes $e^{-iHt}$ is complex; thus, current methods involve decomposing the Hamiltonian into a sum of Pauli strings (expressed as $H = \sum_j^L \lambda_j P_j$) which are simpler to exponentiate.
With the approximation given by Trotter Formula \cite{trotter}:
$e^{-iHt} = (\prod_j^L e^{-i\lambda_j P_j\Delta t})^{\frac{t}{\Delta t}} + O(t \Delta t)$,
an approximation of $\ket{\psi(t)}$ can be obtained by synthesize the circuit for $\prod_j^L e^{-i\lambda_j P_j \Delta t}$ and apply it $\frac{t}{\Delta t}$ times to $\ket{\psi(0)}$.
This approach breaks down Hamiltonian simulation to synthesizing the exponential of its Pauli strings $P_0, P_1 \dots, P_{L-1}$, which can be done in a structural way {as discussed below}.
Moreover, the order of the Pauli strings in the Trotter Formula is not specified --- the Pauli strings can be ordered in any of the $L!$ permutations, which means that commuting the exponential of two Pauli strings does not affect the simulation result.

Most real-world Hamiltonians consist of 1- and 2-local Pauli strings. We focus on three widely used physics models of this form: the Transverse Ising model, XY model, and Heisenberg model:
$H_{\textrm{Ising}} = \sum_{\langle i,j\rangle \in E} \gamma_{ij} Z_iZ_j + \sum_k \alpha_k X_k$,
$H_{\textrm{XY}} = \sum_{\langle i,j\rangle \in E}  (\alpha_{ij} X_iX_j + \beta_{ij} Y_iY_j)$,
 $H_{\textrm{Heisenberg}} = \sum_{\langle i,j\rangle \in E}  \allowbreak (\alpha_{ij} X_iX_j + \beta_{ij} Y_iY_j + \gamma_{ij} Z_iZ_j)$,
where $E$ is the set of qubit pairs.
These models share two properties we exploit: (1) 2-local Pauli strings contain homogeneous non-identity operators (e.g., $ZZ$ but not $YZ$); and (2) different Pauli types often act on the same qubit pairs (e.g., $X_iX_j$ implies $Y_iY_j$ in the XY model).

\paragraph{Synthesize the Exponential of 2-Local Pauli Strings}
\label{subsec:synthesize_pauli}

\begin{figure}
    \centering
    \includegraphics[scale=0.75]{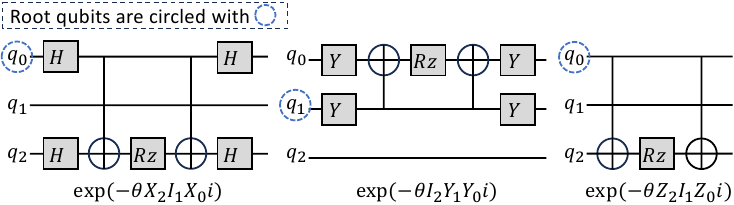}
    \caption{Synthesizing the exponential of 2-local Pauli strings.}
    \label{fig:pauli_string_synthesizing}
\end{figure}

The synthesis of 2-local Pauli strings into circuits follows the structure in \Cref{fig:pauli_string_synthesizing}. The circuit is symmetric: each end applies single-qubit gates to qubits with non-identity operators—$H$ for $X$, $Y$ for $Y$, and identity for $Z$. Inside, two symmetric CNOTs are placed, with one qubit designated as the \textit{root} (control) and the other as the target, allowing flexibility in synthesis. An $R_z$ gate, parameterized by the Pauli string’s exponential factor, is inserted between the CNOTs.

\paragraph{Quantum Architectures and Chiplets}

In the realm of quantum computing, there are three main architectures: monolithic, distributed, and chiplet. The monolithic architecture
faces scalability challenges in increasing qubit counts to meet the demands of complex algorithms \cite{grover1996fast, shor1999polynomial, hhl} and issues with frequency collisions and fabrication yields as the scale increases \cite{smith2022scaling}. The distributed architecture, linking multiple quantum nodes via interconnect systems\cite{han2021microwave, lambert2020coherent}, addresses some scalability issues but suffers from noisier and slower inter-node connections \cite{ang2022architectures}.
Chiplet architecture integrates identical chiplets with a host chip to form multi-chip modules (MCMs) \cite{smith2022scaling, chow2021quantum, laracuente2022modeling}, reducing frequency-collision risks of larger devices and improving scalability and efficiency. Unlike distributed architectures that assume all-to-all connectivity, chiplet architectures feature less noisy inter-chip links restricted to adjacent chiplets. We refer to this arrangement as the {\em chiplet topology graph}.

\paragraph{Highway}
\label{subsec:highway}

MECH \cite{highway} introduces the concept of a \emph{highway} to enable efficient qubit communication in circuit compilation on chiplet architectures, formally defined as follows:

\begin{definition}[Highway]
 A highway is a quantum state entangling multiple qubits. Suppose the number of qubits involved is $m$, a highway is
 \[\frac{1}{\sqrt{2}}(\ket{0}^{\otimes m} + \ket{1}^{\otimes m})\]
\end{definition}

Highways offer two advantages: (1) long-distance communication at constant circuit depth, unlike the linear depth growth of conventional routing; and (2) efficient multi-target controlled gates (\emph{highway gates} \cite{highway}), where $O(n)$ CNOTs with a shared control are implemented in $O(1)$ depth instead of $O(n)$ (see \Cref{fig:multi_target}). The formal definition of a highway gate is as follows:

\begin{definition}[Highway gate]
A highway gate is defined as a sequence of CNOT operations sharing the same control qubit but with distinct target qubits, applied sequentially with no other gates interleaved.
\end{definition}

Highways build on generalized GHZ and cat-like states \cite{generalized_ghz}. An $m$-qubit GHZ state is $\tfrac{1}{\sqrt{2}}(\ket{0}^{\otimes m} + \ket{1}^{\otimes m})$, while a cat-like state is $\alpha\ket{0}^{\otimes m} + \beta\ket{1}^{\otimes m}$.
\cite{generalized_ghz} introduces a \emph{cat entangler} to create a cat-like state from a single qubit $\alpha \ket{0} + \beta\ket{1}$ and a generalized GHZ state, using local operations and classical communication. This cat-like state can control multiple targets. After entanglement, a \emph{cat disentangler}, shown in \Cref{fig:cat_entangler}, is used to disentangle the control qubit from the cat-like state.
While both cat entangler and disentangler are efficiently implementable with constant depth, preparing GHZ states efficiently is challenging.
MECH \cite{highway} proposes a fast method using mid-circuit measurement \cite{hua2023caqr} (\Cref{fig:fast_preparation}).
For long-distance CNOT gates with constant depth, a generalized GHZ state is first constructed along the path between two qubits, then entangled with the control qubit to form a cat state for controlling the remote target qubit. Similar principles apply to multi-target controlled gates.
Highway construction requires \emph{ancilla qubits}, auxiliary qubits distinct from \emph{data qubits}, which remain in $\ket{0}$ when idle.

\subsection{Related works}
\label{subsec:related_work}

\begin{figure*}
  \centering
  \subfloat[] {
    \includegraphics[width=0.4\textwidth]{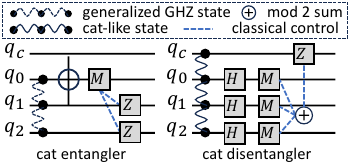}
    \label{fig:cat_entangler}
  }
  \\
  \subfloat[] {
        \includegraphics[width=0.43\textwidth]{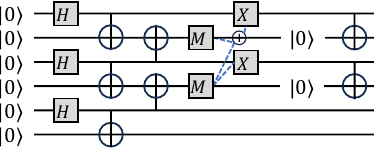}
        \label{fig:fast_preparation}
  }
  \subfloat[] {
        \includegraphics[width=0.35\textwidth]{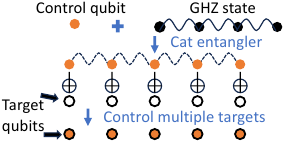}
        \label{fig:multi_target}
  }
  \caption{(a) The mechanism of the cat entangler and cat disentangler. (b) The fast preparation of generalized GHZ states. (c) Multiple target being controlled in parallel by a cat-like state.}
  \label{fig:pauli_string_synthesis}
\end{figure*}
Quantum compilation aims at mitigating the issues, including noisy operation, short qubit decoherence time and sparse physical qubit connectivity, of near term quantum devices.

\paragraph{General-purpose compilers}
General-purpose quantum compilers \cite{qiskit, tket, li2023quarl, xu2022quartz, xu2023synthesizing, kissinger2019pyzx} optimize circuits to reduce operation count and depth, thereby improving fidelity. To handle limited qubit connectivity, they insert SWAP gates to move states between connected qubits—a process called {\em qubit routing}. Routing efficiency heavily depends on the initial mapping of logical to physical qubits, making the mapper a critical component of these compilers.

\paragraph{Compilers for the chiplet architecture}

The chiplet architecture introduces unique compilation challenges. Its scalability enlarges device size, increasing communication costs between distant qubits through added depth and SWAPs. Cross-chiplet couplings are also noisier than on-chip links, making it essential to minimize their use. Moreover, unlike distributed architectures with all-to-all connectivity, compilation must account for the chiplet topology graph. MECH \cite{highway} addresses these challenges by using highways for cross-chiplet communication, reducing routing costs and reliance on error-prone couplings.

\paragraph{Domain-specific compilers for Hamiltonian simulation}
Domain-specific compilers \cite{2qan, paulihedral, li2021software} exploit Hamiltonian semantics—particularly Pauli string commutativity—to improve performance. However, they ignore the heterogeneity of chiplet architectures and cannot scale to problem sizes suited for them. For example, 2QAN’s scheduling complexity is $O(m^4)$, where $m$ is the number of Pauli strings (at least linear in qubit count), severely limiting scalability.

\section{Motivation}

To motivate \sys’s design, this section presents new findings about highways that make them particularly effective for 2-local Hamiltonian simulation.

\subsection{2-Local Pauli String Aggregation}
\label{subsec:pauli_string_aggregate}

\begin{figure*}[t]
    \centering
    \subfloat[Commutativity of CNOTs.] {
        \includegraphics[width=0.3\textwidth]{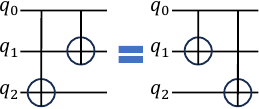}
        \label{fig:cnot_commute}
      }
    \\
      \subfloat[Example of Pauli string aggregation.] {
        \includegraphics[width=0.8\textwidth]{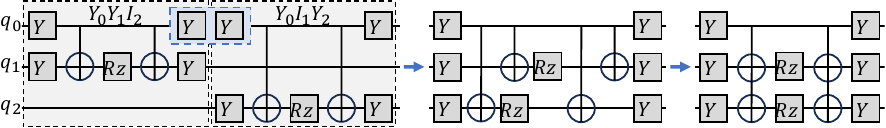}
        }
    \caption{
    An example of aggregating two Pauli strings, $Y_0Y_1I_2$ and $Y_0I_1Y_2$, into two highway gates with three layers of single-qubit gates. The original circuit contains the two Pauli strings; first, the $Y$ gates in the blue dotted box cancel. Next, the rewrite rule in (a) groups the CNOTs, which are then aggregated into highway gates when they share a control qubit.
    }
    \label{fig:aggregation_example}
\end{figure*}

As stated in \Cref{subsec:highway}, a core use case of highway is highway gates.
Intuitively, the compiler should find as many CNOT gates sharing the same control qubit as possible and aggregate them into a highway gate.
MECH \cite{highway} uses circuit rewriting rules to maximize this sharing.

\begin{figure}
    \centering
    \includegraphics[scale=0.80]{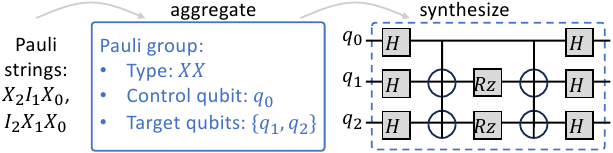}
    \caption{Aggregating two Pauli strings into a Pauli set.}
    \label{fig:pauli_set}
\end{figure}

The structure of 2-local Hamiltonian simulation enables {\em systematic aggregation} of highway gates. When multiple 2-local Pauli strings of the same type share a qubit, they can be aggregated by choosing that qubit as the root. For example, \Cref{fig:aggregation_example} shows how \sys aggregates $Y_0Y_1I_2$ and $Y_0I_1Y_2$, both $YY$ strings sharing qubit $q_0$ as the root. This approach generalizes to $XX$ and $ZZ$ strings, as well as to any set of $n>2$ strings of the same type sharing a qubit.

Because the Trotter formula holds for all permutations of Pauli strings (\Cref{sec:background}), all 2-local Pauli strings of the same type sharing a qubit can be aggregated. To capture this, we introduce the notion of a {\em Pauli set}, representing the aggregation of multiple 2-local Pauli strings. As shown in \Cref{fig:pauli_set}, each Pauli set is defined by three attributes: the Pauli type, a control (root) qubit, and the set of target qubits. Each Pauli string $p$ also has a parameter $\theta$ from its exponential $e^{-i\theta p}$, but since $\theta$ does not affect the structure of the set, we omit it in system design while tracking it in \sys’s implementation. \sys opportunistically aggregates Pauli strings into Pauli sets to optimize circuits. Given a set of 2-local Pauli strings, multiple aggregation plans may exist. For example, ${Z_2Z_1, Z_2Z_0, Z_1Z_0}$ can be aggregated in three ways, each forming two Pauli sets. \Cref{subsec:schedule} describes how \sys selects an aggregation plan.

\subsection{Highway Reuse}
\label{subsec:highway_reuse}

\begin{figure*}
\centering
  \subfloat[] {
    \includegraphics[width=0.13\textwidth]{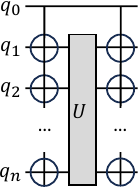}
    \label{fig:thm1}
  }
  \subfloat[] {
    \includegraphics[width=0.26\textwidth]{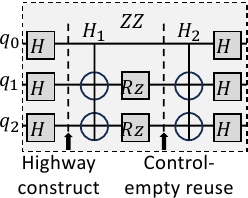}
    \label{fig:zz_example}
  }
  \subfloat[] {
    \includegraphics[width=0.13\textwidth]{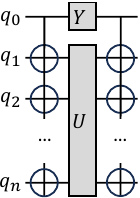}
    \label{fig:zz_}
  }
  \\
  \subfloat[] {
    \includegraphics[width=0.75\textwidth]{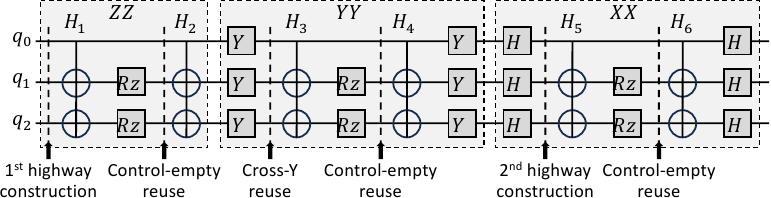}
    \label{fig:heisenberg_example}
  }
  \caption{Different ways to reuse a highway and examples. (a) Control-empty highway reuse. (b) An example of control-empty highway reuse in $ZZ$ Pauli sets. (c) Cross-Y highway reuse. (d) An example of cross-Y highway reuse in a Heisenberg model. In both (a) and (c), the two highway gates share the same control qubit (i.e., $q_0$) and target qubits (i.e., $q_1,...,q_n$).}
  \label{fig:theorem_1}
\end{figure*}

Another key insight motivating \sys is that {\em highways are reusable}. In this section, we present two reuse cases specific to 2-local Hamiltonian simulation.

\paragraph{Case 1: control-empty highway reuse.} The first type of highway reuse identifies scenarios where two highway gates with the same control and target qubits are separated by gates applying {\em only} to the target qubits (i.e., no gates performing on the control qubit between the two highway gates).
As shown in \Cref{fig:thm1}, $q_0$ is the control for both highways, separated by a unitary $U$ on targets $q_1, q_2,\ldots,q_n$.
This pattern is very common in 2-local Hamiltonian simulation: within a Pauli set, two identical highway gates are separated by $R_z$ layers on the targets.
For example, \Cref{fig:zz_example} shows control-empty reuse in a $ZZ$ Pauli set, where a highway built for $H_1$ is reused for $H_2$.
Correctness is established in \Cref{thm:1}.

\begin{theorem}[Control-empty highway reuse]
\label{thm:1}
    For two highway gates sharing the same control and target qubits, if they are separated by gates applied to {\em only} the target qubits, the highway established for the first highway gate can be directly reused by the second highway gate without reconstruction.
\end{theorem}

\begin{proof}
Let us consider a quantum system comprising a set of qubits, denoted by $Q$. Within this set, $A \subset Q$ represents the ancillary qubits used to construct quantum highways, and $D = Q \setminus A$ consists of the data qubits, where $|A| = m$ and $|D| = n$. We denote $q_0$ as the control qubit, $T_1 \subseteq D \setminus \{q_0\}$ as the target qubits for the first highway gate, and $T_2 \subseteq D \setminus \{q_0\}$ for the second highway gate. An arbitrary unitary operation $U$ is applied to the qubits in $D \setminus \{q_0\}$ between the execution of these two highway gates.

Previous works \cite{highway, generalized_ghz} have established that the final state resultant from executing two highway gates with highways is identical to the outcome achieved through their local execution. Thus, to affirm the efficacy of highway reuse, it suffices to demonstrate that the state resulting from highway reuse mirrors that derived from local gate applications.

Assume the initial state of the data qubits is given by
{
\small
\begin{align*}
\small
\sum_{x \in \{0, 1\}^{n-1}} \alpha_x \ket{0}\ket{x} + \sum_{x \in \{0, 1\}^{n-1}} \beta_x \ket{1}\ket{x}
\end{align*}
}
, with the first qubit being the control qubit. Upon applying the first highway gate, the system evolves to
{
\small
\begin{align*}\small
\sum_{x \in \{0, 1\}^{n-1}} \alpha_x \ket{0}\ket{x} + \sum_{x \in \{0, 1\}^{n-1}} \beta_x \ket{1}O_1\ket{x}
\end{align*}
}
, where $O_1$ is defined as
{
\small
\begin{align*}\small
  O_1 = \bigotimes_{i=1}^{n} O_i \quad \text{where} \quad O_i =
\begin{cases}
X & \text{if } q_i \in T_1 \\
I & \text{otherwise}
\end{cases}.
\end{align*}
}
Following the application of $U$ and the second highway gate, the state becomes:
{
\small
\begin{align*}\small
\sum_{x \in \{0, 1\}^{n-1}} \alpha_x \ket{0}U\ket{x} + \sum_{x \in \{0, 1\}^{n-1}} \beta_x \ket{1}O_2UO_1\ket{x},
\end{align*}
}
where $O_2$ is similarly defined by:
{
\small
\begin{align*}
  O_2 = \bigotimes_{i=1}^{n} O_i \quad \text{where} \quad O_i =
\begin{cases}
X & \text{if } q_i \in T_2 \\
I & \text{otherwise}
\end{cases}.
\end{align*}
}

Upon the construction of the quantum highway and its initial entanglement with the control qubit, the system is prepared in a superposition state that reflects the entanglement between the ancillary qubits and the control qubit alongside the data qubits' state
{
\small
\begin{align*}\small
\sum_{x \in \{0, 1\}^{n-1}} \alpha_x \ket{0^{\otimes m+1}}\ket{x} + \sum_{x \in \{0, 1\}^{n-1}} \beta_x \ket{1 ^{\otimes m+1}}\ket{x}
\end{align*}
}
Applying the first highway gate based on the control qubit's state modifies this to
{
\small
\begin{align*}\small
\sum_{x \in \{0, 1\}^{n-1}} \alpha_x \ket{0^{\otimes m+1}}\ket{x} + \sum_{x \in \{0, 1\}^{n-1}} \beta_x \ket{1 ^{\otimes m+1}}O_1\ket{x}.
\end{align*}
}
Following the unitary operation $U$ and the second highway gate, the system evolves to
{
\small
\begin{align*}\small
\sum_{x \in \{0, 1\}^{n-1}} \alpha_x \ket{0^{\otimes m+1}}U\ket{x} +  \sum_{x \in \{0, 1\}^{n-1}} \beta_x \ket{1 ^{\otimes m+1}}O_2UO_1\ket{x}
\end{align*}
}
The disentangler circuit then resets the ancillary qubits, yielding the final state:
{
\small
\begin{align*}
   \ket{0^{\otimes m}}\otimes (\sum_{x \in \{0, 1\}^{n-1}} \alpha_x \ket{0}U\ket{x} +  \sum_{x \in \{0, 1\}^{n-1}} \beta_x \ket{1 }O_2UO_1\ket{x})
\end{align*}
}
demonstrating that highway reuse achieves the same quantum state as local operations, thereby validating its efficiency and equivalence.
\end{proof}

\paragraph{Case 2: cross-Y highway reuse.}
The second type of highway reuse identifies scenarios where two highway gates have the same control and target qubits and there is {\em exactly one} $Y$ gate applied to the control qubit between the two highway gates, as illustrated in \Cref{fig:heisenberg_example}.
The only difference from control-empty reuse is this intervening $Y$.
This pattern commonly appears in the Heisenberg model when a $ZZ$ Pauli set is followed by a $YY$ set.
To reuse the highway, we apply a $Y$ to the control and an $X$ to all other qubits in the cat-like state.
\Cref{fig:heisenberg_example} illustrates this with three Pauli sets ($ZZ$, $YY$, $XX$): a highway built for $H_1$ is reused by $H_2$ (control-empty reuse), then extended to $H_3$ via cross-Y reuse and further to $H_4$. Reuse ends at $H_5$ due to an $H$ gate on the control.

\begin{theorem}[Cross-Y highway reuse]
\label{thm:2}
    For two highway gates sharing the same control and target qubits, if there is {\em exactly one} $Y$ gate on the control qubit between the two highway gates, are separated by gates applied to {\em only} the target qubits, the highway established for the first highway gate can be reused for the second highway gate by applying a $Y$ gate on the control qubit and an $X$ gate on all highway qubits.
\end{theorem}

\begin{proof}
Adopting the notation from the proof of \Cref{thm:1}, we analyze the system's evolution through both non-highway and highway scenarios to demonstrate equivalence in the resulting quantum states.
Initially, without using the highway, the system state after applying the first highway gate is given by
{
\small
\begin{align*}
\sum_{x \in \{0, 1\}^{n-1}} \alpha_x \ket{0}\ket{x} + \sum_{x \in \{0, 1\}^{n-1}} \beta_x \ket{1}O_1\ket{x}
\end{align*}
}
Following the application of a $Y$ gate to the control qubit, the state transforms to
{
\small
\begin{align*}
\sum_{x \in \{0, 1\}^{n-1}} i\alpha_x \ket{1}\ket{x} - \sum_{x \in \{0, 1\}^{n-1}} i\beta_x \ket{0}O_1\ket{x}
\end{align*}
}
, since $Y\ket{0} = i\ket{1} \text{and } Y\ket{1} = -i\ket{0}$.
Next, applying the unitary gate $U$ to non-control qubits results in
{
\small
\begin{align*}
\sum_{x \in \{0, 1\}^{n-1}} i\alpha_x \ket{1}U\ket{x} - \sum_{x \in \{0, 1\}^{n-1}} i\beta_x \ket{0}UO_1\ket{x}
\end{align*}
}
The application of the second highway gate modifies the state to:
{\small
\begin{align*}
\sum_{x \in \{0, 1\}^{n-1}} i\alpha_x \ket{1}O_2U\ket{x} - \sum_{x \in \{0, 1\}^{n-1}} i\beta_x \ket{0}UO_1\ket{x}.
\end{align*}
}
In the highway reuse scenario, after the first highway gate, the state is
{
\small
\begin{align*}
\sum_{x \in \{0, 1\}^{n-1}} \alpha_x \ket{0^{\otimes m+1}}\ket{x} + \sum_{x \in \{0, 1\}^{n-1}} \beta_x \ket{1^{\otimes m+1}}O_1\ket{x}
\end{align*}
}
Applying a $Y$ gate to the control and $X$ gates to all highway qubits yields
{
\small
\begin{align*}
\sum_{x \in \{0, 1\}^{n-1}} i\alpha_x \ket{1^{\otimes m+1}}\ket{x} - \sum_{x \in \{0, 1\}^{n-1}} i\beta_x \ket{0^{\otimes m+1}}O_1\ket{x}
\end{align*}
}
The subsequent application of $U$ and the second highway gate leads to
{
\small
\begin{align*}
\sum_{x \in \{0, 1\}^{n-1}} i\alpha_x \ket{1^{\otimes m+1}}O_2U\ket{x} - \sum_{x \in \{0, 1\}^{n-1}} i\beta_x \ket{0^{\otimes m+1}}UO_1\ket{x}
\end{align*}
}
Finally, disentangling the control qubit from the highway qubits results in
{
\small
\begin{align*}
\ket{0^{\otimes m}}\otimes(\sum_{x \in \{0, 1\}^{n-1}} i\alpha_x \ket{1}O_2U\ket{x} - \sum_{x \in \{0, 1\}^{n-1}} i\beta_x \ket{0}UO_1\ket{x}).
\end{align*}
}
confirming the equivalence of the final states obtained through highway reuse and local gate application.
\end{proof}

\begin{figure}
    \centering
    \includegraphics[scale=0.78]{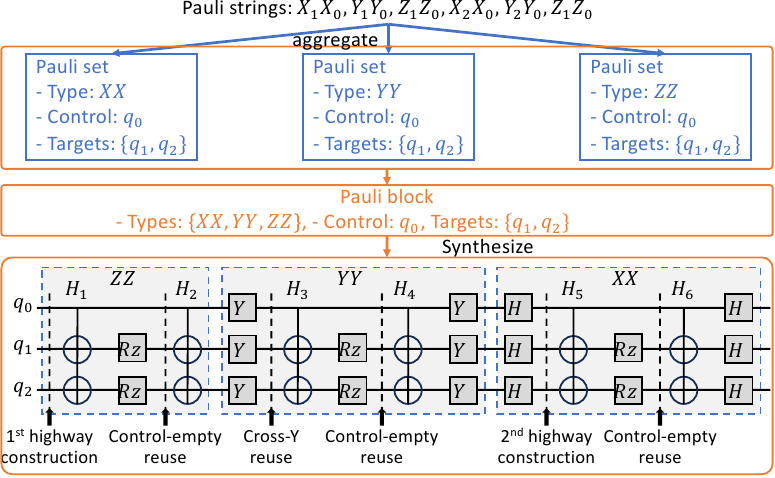}
    \caption{Aggregating three Pauli sets into a Pauli block and synthesizing a circuit by reusing highways. $H_1$-$H_4$ reuse the same highway, so do $H_5$-$H_6$.}
    \label{fig:highway_block}
    \end{figure}

As noted in \Cref{sec:background}, Pauli strings of different types act on the same qubits, allowing $XX$, $YY$, and $ZZ$ strings to be aggregated into their respective Pauli sets with the same aggregation plan. Leveraging \Cref{thm:2}, a $YY$ set can be placed after a $ZZ$ set on the same control–target pair, enabling cross-Y highway reuse. Accordingly, \sys aggregates $ZZ$ and $YY$ strings together and orders $YY$ sets after $ZZ$ sets to maximize reuse.

Placing Pauli sets with the same control and target qubits together also improves {\em locality}.
Executing a Pauli set requires routing its control and targets to the highway; grouping sets with identical qubits reduces this overhead, since the qubits remain near the highway after the first set.
To exploit this, \sys adopts a unified aggregation plan for $XX$, $YY$, and $ZZ$ sets, ordering them as $ZZ$, $YY$, $XX$ to maximize highway reuse and minimize routing cost.
We define a {\em Pauli block} as a group of Pauli sets with the same control and targets but different Pauli types.
\Cref{fig:highway_block} illustrates this: three Pauli sets share one control and the same targets, yielding six highway gates. After the first highway is built for $H_1$, gates $H_2$–$H_4$ reuse it; the highway is then disentangled, a second one is constructed, and $H_5$–$H_6$ reuse it via control-empty reuse.

\section{The Design of \sys}
\label{sec:overview}

\begin{figure*}
    \centering
    \includegraphics[scale=0.67]{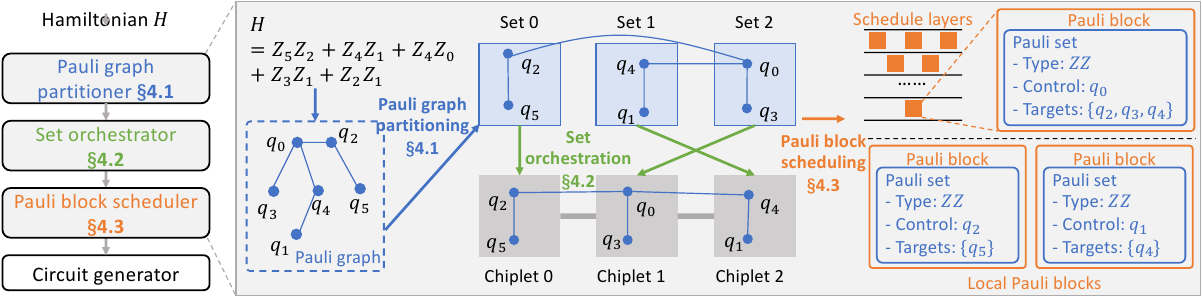}
    \caption{
    Overview of \sys. The example illustrates three components: (1) the Pauli graph partitioner builds and partitions the Pauli graph into three sets; (2) the set orchestrator maps these sets to chiplets to minimize communication; and (3) the Pauli block scheduler aggregates Pauli strings into blocks and arranges them into schedule layers. Although multiple layers may exist in general, the example shows a single layer with one block.
    }
    \label{fig:overview}
    \end{figure*}

Building on these motivations, we propose \sys, a compiler for 2-local Hamiltonian simulation on quantum chiplets. \sys introduces the {\em Pauli graph}, which captures qubit interactions in the Hamiltonian. An overview is shown in \Cref{fig:overview}. The compilation pipeline has five stages: (1) the {\em Pauli graph partitioner}, which assigns qubits to Pauli sets to minimize cross-set strings (\Cref{subsec:graph_partition}); and (2) the {\em set orchestrator}, which maps each Pauli set to a chiplet while minimizing inter-chiplet communication (\Cref{subsec:chiplet_mapping}). Both stages are formalized as constrained optimization problems and solved optimally via integer linear programming (ILP).
Next, the {\em Pauli block scheduler} aggregates Pauli strings within each set into Pauli blocks and arranges them into schedule layers, where multiple blocks run concurrently on chiplets using highways (\Cref{subsec:schedule}). Finally, the {\em intra-chiplet mapper} and {\em circuit generator} (\Cref{sec:circuit_generator}) synthesize these layers into circuits. \Cref{sec:discussion} explains the rationale behind this component ordering.

\subsection{Pauli Graph Partitioner}
\label{subsec:graph_partition}

A central challenge in compiling 2-local Hamiltonians is {\em qubit mapping}, which assigns each logical qubit to a physical qubit. While prior work~\cite{sabre, molavi2022qubit, murali2019noise} emphasizes its importance for monolithic architectures, the challenge is amplified in chiplet systems with long-distance interactions and heterogeneity. Long-range communication increases both circuit depth and gate count, and cross-chiplet couplings are noisier than intra-chiplet ones. MECH~\cite{highway} mitigates these costs using highways, but optimization challenges remain: highways achieve constant depth, yet their CNOT and measurement counts grow linearly with distance, they consume ancilla qubits, and long highways may block other communications. Existing solver-\cite{murali2019noise} and heuristic-based\cite{sabre} mapping methods either ignore heterogeneity and highways or fail to scale to large problems.

\sys employs a {\em hierarchical qubit mapping} strategy: first assigning qubits to chiplets, then mapping each chiplet’s qubits to its physical registers. Once placed, a qubit remains on its chiplet, reducing cross-chiplet routing and coupling. As shown in \Cref{fig:overview}, the {\em Pauli graph partitioner} and {\em set orchestrator} handle chiplet-level assignment, while the {\em intra-chiplet mapper} performs local mapping.

\begin{definition}[Pauli graph]
Given a Hamiltonian on $n$ qubits with Pauli strings $H_1, \ldots, H_m$, the \emph{Pauli graph} is the undirected graph $\mathcal{G} = \langle\mathcal{V}, \mathcal{E}\rangle$ where $\mathcal{V} = \{v_1, \ldots, v_n\}$ represents the qubits, and $(v_i, v_j) \in \mathcal{E}$ if and only if some $H_k$ acts non-trivially on both qubits $i$ and $j$.
\end{definition}

The Pauli graph partitioner groups qubits into sets by minimizing Pauli strings that span multiple sets. The {\em set orchestrator} then maps each set to a chiplet, placing communication-intensive sets on nearby chiplets to reduce inter-chiplet distance. \Cref{fig:overview} illustrates their cooperation. For example, with chiplets holding at most two qubits, the partitioner divides the input Hamiltonian $H$ into three sets. Since set 2 communicates with both sets 0 and 1 (which do not communicate with each other), \sys maps sets 0, 1, and 2 to chiplets 0, 2, and 1, respectively, to respect connectivity and minimize communication cost.

\paragraph{Problem formulation.}
Let $\mathcal{G} = \langle\mathcal{V}, \mathcal{E}\rangle$ be the Pauli graph of a Hamiltonian $H$ on $n$ qubits, $N_c$ the number of chiplets, and $N_q$ the maximum qubits per chiplet. \sys partitions $\mathcal{G}$ into $N_c$ subgraphs, each with at most $N_q$ nodes, while minimizing inter-subgraph edges. This Pauli graph partitioning is formulated as a constrained optimization problem and solved using binary integer linear programming (ILP) with an off-the-shelf solver. We next present the ILP formulation.

\paragraph{Variables and constraints.}
Let $X$ be a $|V| \times N_c$ binary matrix where $X_{ij} = 1$ iff node $v_i$ is assigned to set $j$. Similarly, let $Y$ be a $|E| \times N_c$ binary matrix where $Y_{ij} = 1$ iff edge $e_i = (v_k, v_l)$ lies entirely within set $j$ (i.e., both $v_k$ and $v_l$ are assigned to $j$).
The formulation enforces the following constraints: (1) each node is assigned to exactly one set, $\forall i \in \{0,\dots,n-1\}, \sum_{j=0}^{N_c-1} X_{ij} = 1$; (2) each set contains at most $N_q$ nodes, $\forall j \in \{0,\dots,N_c-1\}, \sum_{i=0}^{n-1} X_{ij} \le N_q$; and (3) $Y$ depends on $X$, with $Y_{ij} = X_{kj} \wedge X_{lj}$ for edge $e_i=(v_k,v_l)$. Since this relation is non-linear, \sys linearizes it as $Y_{ij} \le X_{kj}$, $Y_{ij} \le X_{lj}$, and $Y_{ij} \ge X_{kj} + X_{lj} - 1$.

\paragraph{Objective}
The optimization aims to maximize \[C_{QG}(X, Y) =  \sum_{i=0}^{|E|-1} \sum_{j=0}^{N_c-1} Y_{ij}\],
the total number of intra-set edges, i.e., edges whose endpoints are assigned to the same set.

\paragraph{Scalability analysis}
The formulation has $O(N_c n^2)$ variables and constraints. ILP solvers work well for small $n$, but quickly become intractable as $n$ grows. To handle larger instances, we implement a Pauli graph partitioner based on the multi-level graph partitioning algorithm in METIS \cite{karypis1997metis}. Our evaluation compares the performance of the ILP solver and the METIS heuristic.

\subsection{The Set Orchestrator}
\label{subsec:chiplet_mapping}
After partitioning the Pauli graph into sets, the {\em set orchestrator} assigns each set to a chiplet while minimizing inter-chiplet communication cost. \sys derives the optimal orchestration strategy via an ILP formulation.

\paragraph{Problem formulation.}
\sys models set orchestration as a quadratic assignment problem. Given $N_c$ Pauli sets $\{g_0, g_1, \dots, g_{N_c-1}\}$ and $N_c$ chiplets, let $W$ be the communication intensity matrix, where $W_{ij}$ is the communication between sets $g_i$ and $g_j$, and let $D$ be the distance matrix, where $D_{ij}$ is the shortest-path distance between chiplets $i$ and $j$ on the topology graph. The objective is to minimize $\sum_{i,j} W_{ij} D_{\pi(i)\pi(j)}$ over assignments $\pi$ of sets to chiplets. \sys reformulates this task as a binary ILP.

\paragraph{Variables and constraints.}
Let $A$ be a $|N_c| \times |N_c|$ binary matrix where $A_{ij} = 1$ iff set $g_i$ is assigned to chiplet $j$. The constraints are: (1) each set is assigned to exactly one chiplet, $\forall i \in \{0,\dots,N_c-1\}, \sum_{j=0}^{N_c-1} A_{ij} = 1$; and (2) each chiplet hosts exactly one set, $\forall j \in \{0,\dots,N_c-1\}, \sum_{i=0}^{N_c-1} A_{ij} = 1$.

\paragraph{Objective}
The objective minimizes the total communication cost: \[C_{CM}(A) = \tfrac{1}{2}\sum_{i=0}^{N_c-1}\sum_{j=0}^{N_c-1} (A^T W A \odot D)_{ij}\],
where $\odot$ denotes element-wise multiplication. Here, $W$ is the set-level communication matrix, and $A^T W A$ transforms it into the chiplet-level communication matrix. For example, if set $g_i$ is mapped to chiplet $k$ and $g_j$ to chiplet $l$, then $(A^T W A)_{kl} = W_{ij}$. Multiplying this matrix with $D$, the chiplet distance matrix, yields the communication cost between chiplets.

\sys linearizes the problem and solves it with an ILP solver. It introduces a $N_c \times N_c \times N_c \times N_c$ binary tensor $B$, where $B_{ijkl} = 1$ iff communication between sets $g_i$ and $g_j$ is mapped to communication between chiplets $k$ and $l$, i.e., $B_{ijkl} = A_{ik} \wedge A_{jl}$. The $\wedge$ is linearized with three constraints for all $i \ne j, k \ne l$: (1) $B_{ijkl} \le A_{ik}$, (2) $B_{ijkl} \le A_{jl}$, (3) $B_{ijkl} \ge A_{ik} + A_{jl} - 1$. The resulting objective is \[C_{CM}(A,B) = \tfrac{1}{2}\sum_{i=0}^{N_c-1}\sum_{j=0}^{N_c-1}\sum_{k=0}^{N_c-1}\sum_{l=0}^{N_c-1} W_{ij}D_{kl}B_{ijkl}\],
which is minimized by the ILP solver.

\subsection{Pauli Block Scheduler}
\label{subsec:schedule}

\begin{figure}
  \centering
  \subfloat[] {
    \includegraphics[width=0.27\textwidth]{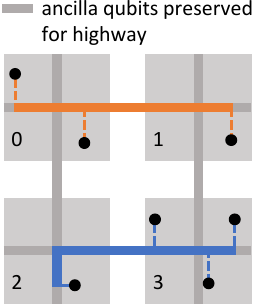}
    \label{fig:parallel_highway}
  }
  \quad
  \subfloat[] {
    \includegraphics[width=0.27\textwidth]{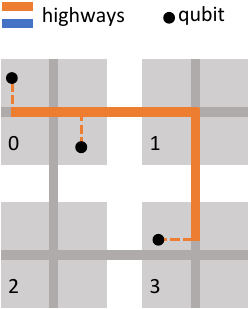}
    \label{fig:routing_chiplet}
  }
  \caption{(a) Parallel Pauli blocks. (b) Data chiplets (0 and 3) and path chiplet (1).}
\end{figure}
\begin{algorithm}[t]
    \caption{Construction of a single schedule layer.}
    {
        \begin{algorithmic}[1]
        \State \textbf{Input:} A logical qubit index to chiplet index mapping $M$, a topology graph for the chiplets $G$, a set of un-aggregated Pauli strings $S$, Pauli block size threshold $t$.
        \State \textbf{Output:} A schedule layer $L$, and the rest of un-aggregated Pauli strings.
        \State $L \leftarrow \Call{List}{\cdot}$
        \State $Q \leftarrow \Call{Queue}{\cdot}$
        \State $Q.\Call{Push}{(0, 1, \dots, N_c-1)}$ \Comment{$N_c$ is the number of chiplets.}
        \While{$!Q.\Call{empty}{ }$}
            \State $C = Q.$\Call{PopFront}{$\cdot$}
            \State $S_C \leftarrow \Call{Set}{\forall p \in S \:\textrm{where}\: M[p.q1] \in C \wedge M[p.q2] \in C}$.
            \State $S_{C\textrm{-cross}} \leftarrow \Call{Set}{\forall p \in S_C \:\textrm{where}\: M[p.q1] \ne M[p.q2]}$.
            \If{$|S_{C\textrm{-cross}}| = 0$} \label{line:intra_begin}
                \For{$c \in C$}
                    \State $S_c \leftarrow \{p|p\in C \wedge M[p.q1] = M[p.q2]=c\}$
                    \State $q_{\textrm{control}} \leftarrow$ the qubit involved in the most $p \in S_c$.
                    \State $S_{\textrm{controlled}} \leftarrow \{p|p \in S_c \wedge q_{\textrm{control}} \in \{p.q1, p.q2\}\}$
                    \State $S_{\textrm{targets}} \leftarrow \{q | \exists p \in S_c, q\in \{p.q1, p.q2\}\wedge q \ne  q_{\textrm{control}}\} $
                    \If{$|S_{\textrm{targets}}| > t$}
                        \State $B \leftarrow \Call{MakePauliBlock}{q_{\textrm{control}}, S_{\textrm{targets}}, S_{\textrm{controlled}}}$
                        \State $S = S - S_{\textrm{controlled}}$ \label{line:intra_end}
                    \EndIf
                \EndFor
            \Else
            \State $q_{\textrm{control}} \leftarrow$ the qubit involved in the most $p \in S_{C\textrm{-cross}}$. \label{line:cross_begin}
            \State $S_{\textrm{controlled}} \leftarrow \{p|p \in S_c \wedge q_{\textrm{control}} \in \{p.q1, p.q2\}\}$
            \State $S_{\textrm{targets}} \leftarrow \{q | \exists p \in S_c, q\in \{p.q1, p.q2\}\wedge q \ne  q_{\textrm{control}}\} $
            \State $B \leftarrow \Call{MakePauliBlock}{q_{\textrm{control}}, S_{\textrm{targets}}, S_{\textrm{controlled}}}$
            \State $L.\Call{append}{B}$
            \State $S = S - S_{\textrm{controlled}}$
            \State $C' \leftarrow C - (B.\textrm{path\_chiplets} \cup B.\textrm{data\_chiplets})$
            \State $G_{C'} \leftarrow G.\Call{SubGraph}{C'}$
            \For{$C'' \in \Call{ConnectedComponents}{G_{C'}}$}
                \State $Q.\Call{Push}{C''}$
            \EndFor \label{line:cross_end}
            \EndIf
        \EndWhile
        \State \textbf{Return:} $L$, $S$
        \end{algorithmic}}
        \label{alg:schedule}
\end{algorithm}

The {\em Pauli block scheduler} aggregates Pauli strings into Pauli blocks and arranges them into {\em schedule layers}, where each layer contains blocks that can run concurrently on chiplets using highways (\Cref{fig:parallel_highway}). For simplicity, \sys permits two blocks to execute in parallel only if they use disjoint chiplets, avoiding the need to track intra-chiplet qubit positions and preventing intersecting highways. Each block’s highway involves two chiplet types: {\em data chiplets}, which hold the block’s qubits, and {\em path chiplets}, intermediate chiplets required to connect disjoint data chiplets. For example, in \Cref{fig:routing_chiplet}, data chiplets 0 and 3 are connected via path chiplet 1. Since multiple connection paths may exist, \sys selects one minimizing the number of path chiplets.
This task corresponds to the Steiner Tree problem, and \sys addresses it by adopting the Mehlhorn algorithm~\cite{mehlhorn1988faster}.

The scheduler aims to (1) minimize the number of schedule layers, reducing overall circuit depth, and (2) decide whether each Pauli block should use a highway or local routing. Highways are preferable when (a) a block spans multiple chiplets—since \sys forbids cross-chiplet routing without highways—or (b) a block within one chiplet exceeds a size threshold, making highways more efficient than local routing. For small blocks (e.g., two qubits), the cost of constructing a highway can outweigh its benefit, so these are executed locally.

\Cref{alg:schedule} outlines how \sys constructs a schedule layer.
The scheduler maintains a queue $Q$ of unoccupied connected components of the chiplet topology graph $G$. In each iteration, it processes one component $C$, aggregating its Pauli strings into blocks until $Q$ is empty. Aggregated blocks are placed in the schedule layer and executed via highways.
During aggregation, \sys's scheduler utilizes highways for 1) cross-chiplet Pauli blocks and 2) intra-chiplet Pauli blocks whose number of target qubits exceeds a threshold $t$.
For each component $C$, the scheduler first collects all cross-chiplet Pauli strings $S_{C\text{-cross}}$. If $S_{C\text{-cross}} \neq \emptyset$ (line \ref{line:cross_begin}-\ref{line:cross_end} in \Cref{alg:schedule}), it selects the control qubit with the most incident strings, aggregates those strings into a Pauli block $B$, and updates $Q$ by removing occupied chiplets from $C$ to form a new subgraph $G_{C’}$. Otherwise, if $S_{C\text{-cross}} = \emptyset$ (lines \ref{line:intra_begin}-\ref{line:intra_end} in \Cref{alg:schedule}), each chiplet in $C$ is processed independently: the largest Pauli block on that chiplet is identified, and if it exceeds a threshold $t$ (where highways become more efficient than local routing), it is aggregated; otherwise, the chiplet is skipped.
The time complexity of constructing a schedule layer is $O(N_c m)$, where $m$ is the number of Pauli strings in the Hamiltonian.

\sys’s scheduler iteratively builds schedule layers (\Cref{alg:schedule}) until no cross-chiplet Pauli strings remain. As shown in \Cref{fig:overview}, after scheduling, some {\em local} Pauli strings may be left outside any layer. Strings acting on the same qubits are aggregated into local Pauli blocks (e.g., $X_3X_1$, $Y_3Y_1$, and $Z_3Z_1$) and executed via local routing.

\subsection{Circuit Generator}
\label{sec:circuit_generator}

As the final step, the {\em circuit generator} synthesizes schedule layers and local Pauli blocks into a quantum circuit, coordinating them to minimize depth. Circuit synthesis proceeds in five passes, described below.

\paragraph{Intra-chiplet mapper.}

Because data qubits move frequently when implementing schedule layers, the effect of intra-chiplet mapping quickly diminishes. \sys therefore optimizes only for short-term benefits: placing the qubits of each local Pauli block close together so entanglement can occur early. For this, \sys applies an off-the-shelf mapper—specifically, the SABRE mapper \cite{sabre} in Qiskit \cite{qiskit}.

\paragraph{Schedule layer reordering}

Since exponentials of Pauli strings commute, schedule layers also commute. \sys’s circuit generator reorders them dynamically, selecting at each step the layer whose data qubits have the smallest average distance to the nearest ancilla qubit.

\paragraph{Generating highway-assisted circuits}

To generate a highway-assisted circuit for a Pauli block, \sys first selects the {\em control} and {\em target entrances}—the highway qubits that couple with the block’s control and target qubits. Entrances are chosen to minimize routing cost: for each, \sys scans all ancilla qubits and picks the one with the earliest {\em available time} (when its last entanglement finishes), updating availability after each assignment. The control entrance is chosen first, then targets in order of increasing distance to ancillas.
Next, \sys determines the highway path. Since data and path chiplets are fixed by the Pauli block scheduler, the circuit generator selects ancillas on these chiplets to form a path, reducing the problem to a Steiner Tree, solved with the Mehlhorn algorithm~\cite{mehlhorn1988faster}. Once the path is set, the Pauli block is implemented: \sys builds the highway (\Cref{subsec:highway}), entangles control and target qubits with their entrances, and reuses the highway within a Pauli set (\Cref{thm:1}). If a $ZZ$ Pauli set is followed by a $YY$ set, the highway is further reused across them (\Cref{thm:2}).

\paragraph{Coordinating local Pauli strings and Pauli blocks}

In each schedule layer, only part of the data qubits are used, allowing concurrent execution of local Pauli blocks on idle qubits. After each layer, \sys performs routing and entanglement for these blocks, but only within the current circuit depth; remaining blocks are deferred to later rounds. For efficiency, \sys applies a {\em permutation-aware} routing strategy, similar to 2QAN~\cite{2qan}, which prioritizes local blocks whose control and target qubits are nearby.

\subsection{Order of the Components}
\label{sec:discussion}

The Pauli block scheduler is placed between the set orchestrator and the circuit generator. This ensures both feasibility and effectiveness. First, whether a Pauli string should use a highway or local routing can only be determined after higher-level mapping decisions. Second, placing the scheduler before the set orchestrator could yield infeasible schedules: blocks grouped as non-conflicting based only on shared qubits may later conflict after mapping, for example through intersecting highway paths. By scheduling after mapping, \sys guarantees that non-conflict constraints are preserved.

\section{Evaluation}

\begin{table}[t]
\centering
\caption{Backend specifications.}
\label{tab:backend}
\setlength{\tabcolsep}{1.1mm}{\fontsize{9}{10}\selectfont
\begin{tabular}{l|c|c|c|c}
\hline
\textbf{Name} & \textbf{\# Data qubit} & \textbf{\# Total qubit} & \textbf{Chiplet size} & \textbf{ Chiplet array} \\
\hline
Backend\_1 & 324 & 408 & 136 & $1\times 3$ \\
Backend\_2 & 440 & 560 & 140 & $2\times 2$ \\
Backend\_3 & 990 & 1260 & 140 & $3\times 3$ \\
\hline
\end{tabular}
}
\end{table}
\begin{figure}
    \centering
    \includegraphics[scale=0.5]{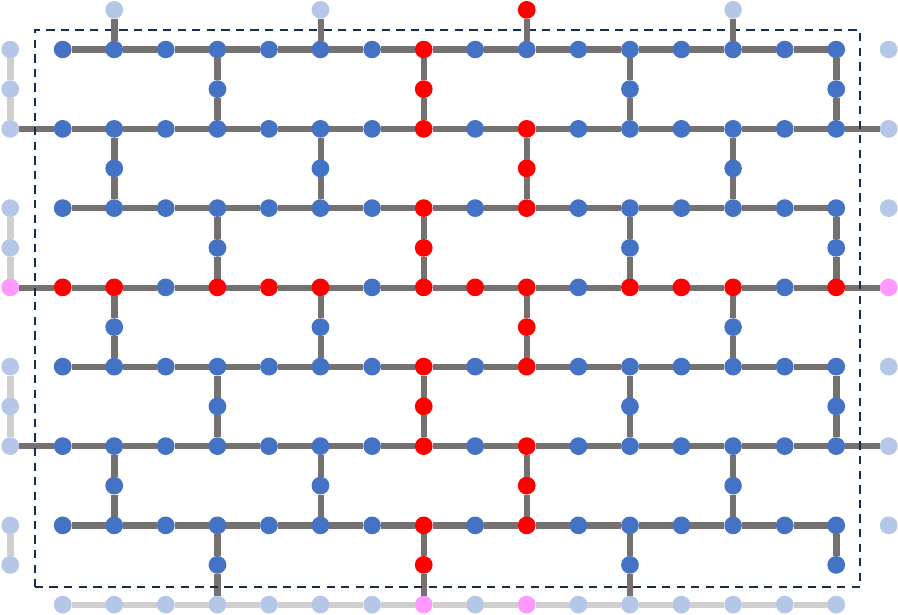}
    \caption{The internal topology of a chplet and cross-chiplet couplings. Red nodes are ancilla qubits preserved for highway. Light-colored nodes are qubits on neigbbor chiplets.}
    \label{fig:chiplet_internal}
    \end{figure}

\begin{figure*}
    \centering
      \subfloat[Compilation time of different compilers.] {
        \includegraphics[width=0.6 \textwidth]{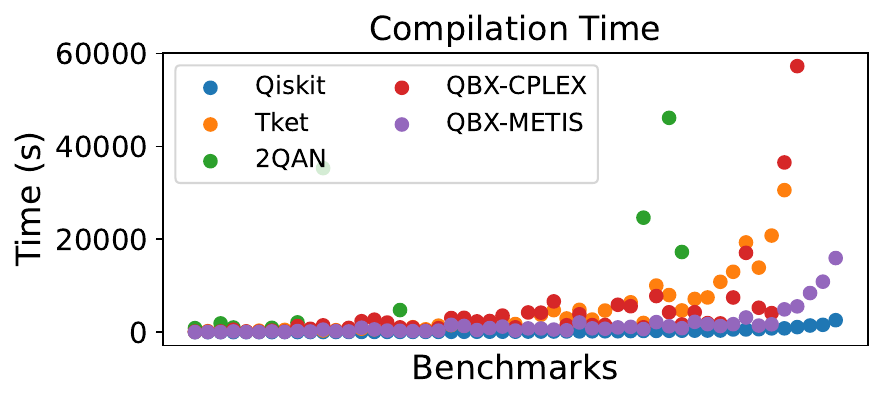}
        \label{fig:time_ratio}
      }
      \\
      \subfloat[Timing of components in \sys-CPLEX.] {
        \includegraphics[width=0.4\textwidth]{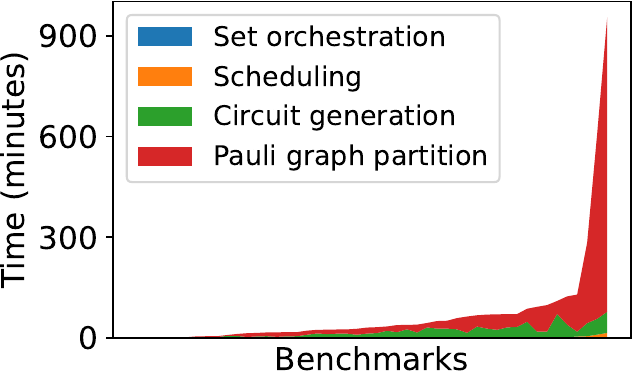}
        \label{fig:time_cplex}
      }
      \quad
      \subfloat[Timing of components in \sys-METIS.] {
        \includegraphics[width=0.4\textwidth]{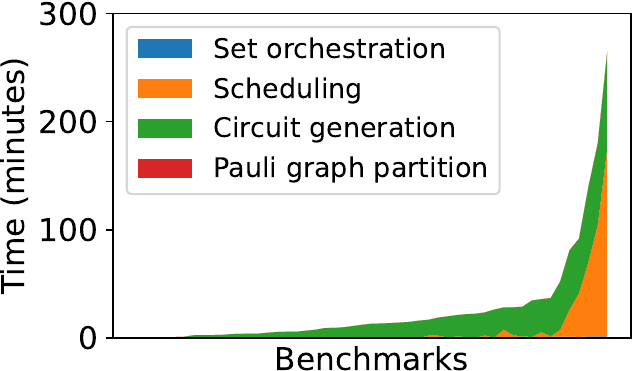}
        \label{fig:time_metis}
      }
    \caption{The running time of different compilers and different components in \sys-CPLEX and \sys-METIS. }
    \end{figure*}

\begin{figure}
\centering
      \subfloat[] {
        \includegraphics[width=0.36\textwidth]{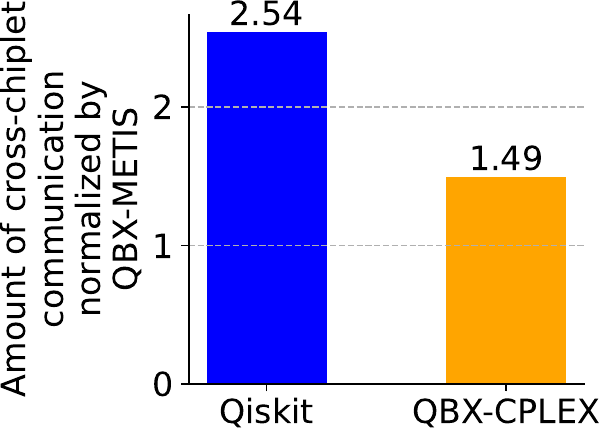}
        \label{fig:comm_amount}
      }
      \subfloat[] {
        \includegraphics[width=0.36\textwidth]{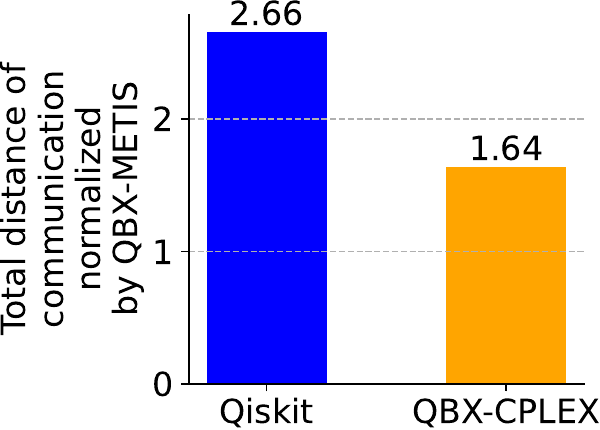}
        \label{fig:comm_dist}
      }
    \caption{Ablation study on \sys's hierarchical mapper.}
    \end{figure}

\begin{figure}
    \centering
      \subfloat[Depth (normalized by Qiskit).] {
        \includegraphics[width=0.4\textwidth]{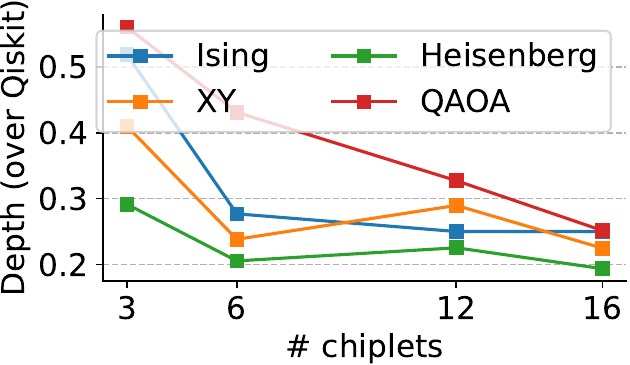}
        \label{fig:scalability_depth}
      }
      \subfloat[Efficient CNOT count (normalized by Qiskit).] {
        \includegraphics[width=0.4\textwidth]{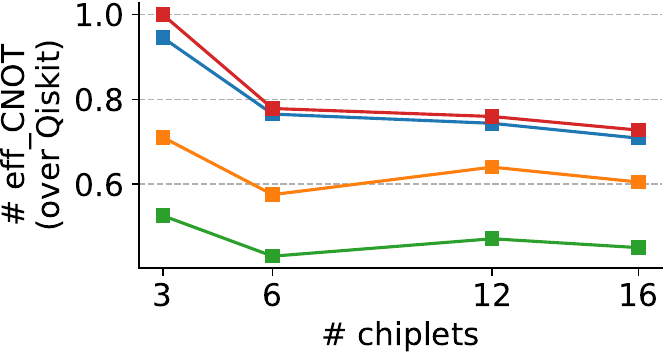}
        \label{fig:scalability_eff_cx}
      }
    \caption{sensitivity study on the influence of the number of chiplets to \sys-METIS's performance.}
    \end{figure}

\begin{table*}[ht]
\centering
\caption{Comparing \sys and existing compilers on the depth and number of effective CNOTs in the output circuits. In the name of the circuits, "Q" refers to QAOA \cite{qaoa}, "H" refers to the Heisenberg model, "reg" refers to regular graph and "rand" refers to ER graph \cite{erdds1959random}.  Best results are bolded. ``-'' indicates that compilation didn't finish in 24 hours for the target task. }
\label{tab:main_results}
\setlength{\tabcolsep}{1.1mm}{\fontsize{7}{8}\selectfont
\begin{tabular}{l|r||r|r|r|r|r|r||r|r|r|r}
\hline
\multirow{2}{5em}{ \textbf{Circuit} } & \multirow{2}{3em}{\bf \# Pauli Strings} & \multicolumn{2}{c|}{\bf Qiskit} & \multicolumn{2}{c|}{\bf Tket} & \multicolumn{2}{c||}{\bf 2QAN} & \multicolumn{2}{c|}{\bf \sys-CPLEX} & \multicolumn{2}{c}{\bf \sys-METIS}\\
\cline{3-12}
& & {\bf depth} & {\bf  eff\_CXs} & {\bf depth} & {\bf  eff\_CXs} &{\bf depth} & {\bf  eff\_CXs} &{\bf depth} & {\bf  eff\_CXs} &{\bf depth} & {\bf  eff\_CXs} \\
\hline
Q\_reg\_324\_16 & 2592 & 6696.7 & 74456.5 & 11851 & 97743.4 & 5906 & \textbf{58960.4} & 4231 & 74742.8 & \textbf{3976} & 75364.7\\
Q\_reg\_324\_32 & 5184 & 10344 & 116877.1 & 18649 & 141084.8 & - & - & \textbf{5105} & \textbf{106702.8} & 5412 & 109509\\
Q\_reg\_324\_64 & 10368 & 14165.7 & 166510.6 & 26140 & 216255.4 & - & - & 6968 & 161782.3 & \textbf{6382} & \textbf{160600.4}\\
Q\_reg\_324\_128 & 20736 & 19517 & \textbf{241406.2} & 42995 & 326729 & - & - & 10711 & 262621 & \textbf{10181} & 257287\\
Q\_rand\_324\_0.1 & 5171 & 10825.7 & 118867.6 & 26539 & 184604.6 & - & - & 4984 & \textbf{106469.5} & \textbf{4815} & 107160.5\\
Q\_rand\_324\_0.2 & 10397 & 14885 & 170675.5 & 27269 & 203604.2 & - & - & 6540 & 162269.5 & \textbf{6475} & \textbf{161673.3}\\
Q\_rand\_324\_0.3 & 15738 & 17142.7 & \textbf{208160.1} & 35001 & 270025.6 & - & - & 8522 & 213589.5 & \textbf{8131} & 210007.9\\
Q\_rand\_324\_0.4 & 20908 & 18504 & \textbf{237955.6} & - & - & - & - & 10656 & 266345.7 & \textbf{10108} & 258192\\
Ising\_1D\_324 & 645 & 284.3 & 6625.6 & 2025 & 18732 & 483 & \textbf{6511.2} & \textbf{138} & 7530.2 & 141 & 7509.3\\
Ising\_2D\_324 & 1646 & 1140 & 22188.6 & 8093 & 67933.2 & \textbf{411} & \textbf{10951.8} & 560 & 23667.3 & 560 & 23667.3\\
Ising\_3D\_324 & 2074& 1383.3 & 28567.9 & 8304 & 78441.2 & \textbf{462} & \textbf{13868.2} & 630 & 27575.9 & 630 & 27575.9\\
XY\_1D\_324 & 1290 & 477 & 11007.5 & 3366 & 31824.4 & - & - & 195 & \textbf{9450.4} & \textbf{178} & 9676.8\\
XY\_2D\_324 & 3292 & 2528.3 & 43004.1 & 15247 & 129304.8 & - & - & \textbf{1002} & \textbf{32314.6} & \textbf{1002} & \textbf{32314.6}\\
XY\_3D\_324 & 4148 & 3312.3 & 55656.7 & 15365 & 135203.6 & - & - & \textbf{1233} & \textbf{39117.8} & 1385 & 39407.3\\
H\_1D\_324 & 1935 & 626.3 & 14658.7 & 4296 & 40872.6 & - & - & \textbf{209} & \textbf{10710.4} & 232 & 10788.6\\
H\_2D\_324 & 4938 & 3848.7 & 63108.1 & 18599 & 163259 & - & - & 1262 & \textbf{35141} & \textbf{1130} & 35468.6\\
H\_3D\_324 & 6222 & 5389 & 83090 & 22624 & 179375.4 & - & - & \textbf{1398} & \textbf{44279.2} & \textbf{1398} & \textbf{44279.2}\\
Q\_reg\_440\_16 & 3520& 10095.3 & 135175.1 & 18610 & 157747.8 & - & - & 5764 & 110480.1 & \textbf{5341} & \textbf{109656.7}\\
Q\_reg\_440\_32 & 7040 & 14106 & 188633.5 & 29774 & 274279.8 & - & - & 8229 & \textbf{157884.4} & \textbf{7444} & 159773\\
Q\_reg\_440\_64 & 14080 & 19400.3 & 267714.2 & 50182 & 429129 & - & - & \textbf{9560} & 233224.2 & 10474 & \textbf{226461.7}\\
Q\_reg\_440\_128 & 28160 & 32101.3 & 387831.8 & - & - & - & - & 15184 & \textbf{359139.8} & \textbf{14601} & 370681.7\\
Q\_rand\_440\_0.1 & 9599 & 16894 & 215377.3 & 37226 & 313570.8 & - & - & 9004 & \textbf{189508.1} & \textbf{8795} & 196354\\
Q\_rand\_440\_0.2 & 19241 & 23522 & 323137.7 & 54480 & 461640.8 & - & - & \textbf{11349} & \textbf{276132.6} & 12817 & 295009.5\\
Q\_rand\_440\_0.3 & 28945& 30546 & 396209.1 & 65437 & 569381 & - & - & 15519 & \textbf{372179.3} & \textbf{14253} & 378858\\
Q\_rand\_440\_0.4 & 38551 & 44233.7 & \textbf{466962.4} & 78995 & 618670 & - & - & 18313 & 478256.7 & \textbf{16311} & 472648.7\\
Ising\_1D\_440 & 877 & 343.3 & 11178.4 & 2178 & 29424.2 & 294 & \textbf{8408.8} & \textbf{120} & 9937.6 & 123 & 9617.5\\
Ising\_2D\_440 & 2177 & 1398.3 & 31987.5 & 9347 & 91196.8 & \textbf{408} & \textbf{15151.4} & 742 & 33242.6 & 554 & 29127.4\\
Ising\_3D\_440 & 3630 & 4116 & 75411.3 & 18823 & 190116.2 & \textbf{902} & \textbf{28643} & 2477 & 62915.3 & 1322 & 55827.2\\
XY\_1D\_440 & 1754 & 538.7 & 17535.1 & 4048 & 44986.8 & - & - & 186 & 12589.4 & \textbf{174} & \textbf{12381}\\
XY\_2D\_440 & 4354& 2800.7 & 57652.7 & 16028 & 166735.8 & - & - & \textbf{918} & 41551.8 & 933 & \textbf{39605.4}\\
XY\_3D\_440 & 7260 & 8888.7 & 148055.5 & 32957 & 291370.2 & - & - & 5252 & 103512.1 & \textbf{2660} & \textbf{80286.8}\\
H\_1D\_440 & 2631 & 881.3 & 23500 & 6848 & 67786.4 & - & - & \textbf{200} & \textbf{14396.4} & 211 & 14411.8\\
H\_2D\_440 & 6531 & 4219 & 87748.6 & 21029 & 209251.4 & - & - & 2051 & 54801.5 & \textbf{1066} & \textbf{43483.6}\\
H\_3D\_440 & 10890 & 12289 & 220924.5 & 40608 & 385503 & - & - & \textbf{2324} & 90806.2 & 3221 & \textbf{87353}\\
Q\_reg\_990\_16 & 7920 & 32332 & 518366.2 & 48584 & 568069.8 & - & - & 13504 & 379438 & \textbf{11269} & \textbf{364683}\\
Q\_reg\_990\_32 & 15840 & 41814 & 781820.5 & 87103 & 945812.6 & - & - & 20042 & \textbf{552707.1} & \textbf{18320} & 559213.2\\
Q\_reg\_990\_64 & 31680 & 50969 & 1081802.1 & 126317 & 1458360.2 & - & - & 30561 & \textbf{804010.2} & \textbf{28359} & 804916.9\\
Q\_reg\_990\_128 & 63360 & 76780.3 & 1531645.9 & - & - & - & - & 41139 & \textbf{1158810.5} & \textbf{39835} & 1170668.1\\
Q\_rand\_990\_0.1 & 48891 & 67074 & 1324073 & 165355 & 1842941.8 & - & - & 36087 & \textbf{1013174.5} & \textbf{35317} & 1022345.7\\
Q\_rand\_990\_0.2 & 97668 & 99520.3 & 1938116.7 & - & - & - & - & - & - & \textbf{50082} & \textbf{1519511.4}\\
Q\_rand\_990\_0.3 & 146631 & 123620.3 & 2395452.1 & - & - & - & - & - & - & \textbf{57177} & \textbf{1950257.6}\\
Q\_rand\_990\_0.4 & 195475 & 156736.7 & 2759370.5 & - & - & - & - & - & - & \textbf{57227} & \textbf{2449045.9}\\
Ising\_1D\_990 & 1977 & 744.3 & 38776.2 & 6028 & 80693.4 & 519 & 26432.6 & \textbf{209} & \textbf{24901.4} & 551 & 39109.8\\
Ising\_2D\_990 & 4927 & 3104.7 & 99440.3 & 27005 & 335886 & 675 & \textbf{41640.8} & 6253 & 208746.9 & \textbf{525} & 64550.5\\
Ising\_3D\_990 & 8227 & 7553.7 & 206696.3 & 50820 & 610385.4 & \textbf{1488} & \textbf{75907.6} & 12175 & 325066.5 & 1627 & 127660.6\\
XY\_1D\_990 & 3954 & 1272.3 & 60105.4 & 9261 & 135214 & - & - & \textbf{280} & \textbf{33699} & 1015 & 58451.8\\
XY\_2D\_990 & 9856 & 4994.3 & 173907.5 & 37896 & 491354.8 & - & - & 13832 & 395610.8 & \textbf{870} & \textbf{89323.2}\\
XY\_3D\_990 & 16454 & 12833.7 & 388667.1 & 77349 & 866650.8 & - & - & 24641 & 617217.2 & \textbf{3045} & \textbf{189103.6}\\
H\_1D\_990 & 5931& 1852.3 & 76722 & 12558 & 189001.2 & - & - & 1324 & \textbf{59791.8} & \textbf{1087} & 62771.2\\
H\_2D\_990 & 14781& 8691.3 & 251162.9 & 52237 & 664488.2 & - & - & 14516 & 398414.6 & \textbf{981} & \textbf{99199}\\
H\_3D\_990 & 24681& 18844 & 568566.5 & 99814 & 1189912.4 & - & - & 27081 & 639511.9 & \textbf{3372} & \textbf{206679.1}\\
\hline
{\bf \begin{tabular}{@{}l@{}}Geo.Mean vs. \\ \sys-METIS \end{tabular}}  & &  2.57 & 1.30 & 10.35 & 2.62 & 1.13 & 0.63 & 1.36 & 1.16 & 1.00 & 1.00 \\
\hline
\end{tabular}
}
\end{table*}

\begin{table}[ht]
\centering
\caption{Comparing \sys and MECH on the depth and number of effective CNOTs in the output circuits.
}
\label{tab:mech}
\setlength{\tabcolsep}{1.1mm}{\fontsize{8.5}{9.5}\selectfont
\begin{tabular}{l|r|r|r|r|r|r}
\hline
\multirow{2}{5em}{ \textbf{Circuit} } &  \multicolumn{2}{c|}{\bf MECH} & \multicolumn{2}{c|}{\bf \sys-CPLEX} & \multicolumn{2}{c}{\bf \sys-METIS}\\
\cline{2-7}
& {\bf depth} & {\bf  eff\_CXs} & {\bf depth} & {\bf  eff\_CXs} &{\bf depth} & {\bf  eff\_CXs} \\
\hline
Q\_reg\_324\_16 & 9592 & 82558.5 & \textbf{3768} & \textbf{73009.2} & 4206 & 80487.9 \\
Q\_reg\_324\_32 & 10417 & 120722.8 & \textbf{5289} & \textbf{105531.6} & 5905 & 123502 \\
Q\_reg\_324\_64 & 11358 & 186580.1 & 7094 & \textbf{160776} & \textbf{6892} & 160909.5 \\
Q\_rand\_324\_0.1 & 10216 & 120624.5 & \textbf{5346} & \textbf{109450.6} & 5574 & 118617.9 \\
Q\_rand\_324\_0.2 & 11375 & 185430.6 & 7058 & 161788 & \textbf{6864} & \textbf{161629.3} \\
Q\_rand\_324\_0.3 & 11921 & 242862 & \textbf{8808} & \textbf{213045.2} & 9367 & 217149 \\
Q\_reg\_440\_16 & 12825 & 126190.7 & 5494 & \textbf{120260.4} & \textbf{5398} & 122220.4 \\
Q\_reg\_440\_32 & 14287 & 184032.2 & \textbf{7622} & \textbf{173367.5} & 7664 & 185896.6 \\
Q\_reg\_440\_64 & 15290 & 276277.2 & 10421 & \textbf{238400.1} & \textbf{9313} & 239063.1 \\
Q\_rand\_440\_0.1 & 14584 & 218050 & 8984 & \textbf{207053.4} & \textbf{8894} & 215363.7 \\
Q\_rand\_440\_0.2 & 15911 & 338972.4 & 12165 & \textbf{288578.9} & \textbf{10985} & 300117 \\
Q\_rand\_440\_0.3 & 16894 & 442429.3 & 13919 & \textbf{373369.6} & \textbf{13491} & 374030.3 \\
\hline
{\bf \begin{tabular}{@{}l@{}}Geo.Mean \\ Ratio to \\ \sys-METIS \end{tabular}}  & 1.69 &  1.08 & 1.00 & 0.96 & 1.00 & 1.00 \\
\hline
\end{tabular}
}
\end{table}
\subsection{Experimental Setup}

\paragraph{Benchmarks}

We evaluate \sys on four 2-local Hamiltonian simulation scenarios: QAOA~\cite{qaoa}, the transverse Ising model~\cite{cipra1987introduction, newell1953theory}, the XY model~\cite{babaev1999nonperturbative, ohta1979xy}, and the Heisenberg model~\cite{heisenberg1985theorie}. For QAOA, we use MAX-CUT Hamiltonians on random $d$-regular graphs (each node has degree $d$) and Erdős–Rényi graphs $G(n,p)$ (a graph of $n$ nodes where each edge appears with probability $p$)~\cite{erdds1959random}. For the physics models, we construct Hamiltonians on 1D, 2D, and 3D lattices with both nearest-neighbor (NN) and next-nearest-neighbor (NNN) interactions. In all benchmarks, Pauli strings are randomly shuffled.

\paragraph{Backends}

\Cref{tab:backend} summarizes the backend specifications, and \Cref{fig:chiplet_internal} shows the topology of a 140-qubit chiplet. All three backends use the heavy-hex layout, widely adopted in superconducting chips for its low error rates~\cite{riel2022quantum}. The smallest backend (408 qubits) appears on IBM’s quantum computing roadmap~\cite{ibm_roadmap}. While our evaluation focuses on heavy-hex, \sys supports arbitrary layouts and flexible ancilla configurations. Unlike MECH, which restricts ancilla assignment to two structured patterns (partially or fully interleaved), \sys allows any mix of consecutive and interleaved assignments. However, ancilla assignment is not yet automated, leaving automation as an open research problem.

\paragraph{Baselines}

We evaluate \sys against several state-of-the-art compilers. Qiskit~\cite{qiskit} and \tket~\cite{tket} serve as general-purpose industry baselines: Qiskit uses transpile at optimization level 3 with three random seeds (averaged), while \tket applies the FullPeepholeOptimization pass with LinePlacement for mapping and LexiRouteRoutingMethod for routing. We also include MECH~\cite{highway}, a chiplet-oriented compiler with highway optimization, and two domain-specific Hamiltonian compilers: Paulihedral~\cite{paulihedral} and 2QAN~\cite{2qan}. For 2QAN, since its tabu search mapper~\cite{tabu} times out, we follow the authors’ recommendation and use Qiskit’s mapper for initial placement.

\paragraph{Metrics}

We evaluate using two metrics: circuit depth and {\em effective CNOT count}. Circuit depth includes only CNOT gates and measurement–reset operations, excluding single-qubit gates and classical communication due to negligible latency. A measurement–reset has depth 2, consistent with device data~\cite{ibm} and prior work~\cite{highway}.
Effective CNOT count~\cite{highway} directly reflects circuit infidelity:$\#CNOT_{\text{eff}} = \#CNOT_{\text{on-chip}} + \tfrac{p_{\text{cross}}}{p_{\text{on}}} \cdot \#CNOT_{\text{cross-chip}} + \tfrac{p_{\text{meas}}}{p_{\text{on}}} \cdot \#\text{measurement}$,
where $p_{\text{on}}$, $p_{\text{cross}}$, $p_{\text{meas}}$ are error rates of on-chip CNOT, cross-chip CNOT, and measurement, respectively. This metric normalizes all error contributions to the infidelity of an on-chip CNOT. We adopt $p_{\text{cross}}/p_{\text{on}} = 7.4$ and $p_{\text{meas}}/p_{\text{on}} = 0.7$ based on prior studies~\cite{gold2021entanglement, highway, kandala2021high} and calibration data~\cite{ibm}.

All experiments ran on a server with dual AMD EPYC 7513 32-core, 64-thread processors and 512 GB DRAM. Baseline compilers were evaluated single-threaded. \sys’s ILP solver used 16 threads, while the METIS-based Pauli graph partitioner ran single-threaded. Each tool was given a 24-hour timeout per circuit.

\subsection{Comparison}
Experiments are conducted to compare \sys with other compilers across 51 circuits. For \sys, we present results using both an ILP solver (specifically, the CPLEX solver), referred to as \sys-CPLEX, and the METIS algorithm for the Pauli graph partitioner, denoted as \sys-METIS. The solver's time limit in the Pauli graph partitioner is set to 10 minutes. All compiled circuits conform to the gate set $\{CNOT, U\}$, with classical communications and measurement-and-reset operations permitted. The results are detailed in \Cref{tab:main_results}.

Compared to generic compilers, \sys achieves substantial gains. \sys-METIS reduces circuit depth by 2.57× over Qiskit~\cite{qiskit} and 10.35× over \tket~\cite{tket}, while lowering effective CNOT count by 1.30× and 2.62×, respectively. Generic compilers miss optimization opportunities by ignoring Pauli string commutativity and further increase errors by neglecting heterogeneity from cross-chiplet couplings.

We also compare \sys with domain-specific compilers, 2QAN~\cite{2qan} and Paulihedral~\cite{paulihedral}. Paulihedral failed on all benchmarks due to its recursive qubit routing, which exceeds Python’s maximum recursion depth on large distances and causes segmentation faults when the limit is reset. 2QAN timed out on most benchmarks; in those completed, \sys achieved 1.13× lower depth, while 2QAN produced fewer effective CNOTs by applying unitary synthesis. As shown in~\cite{vatan2004optimal}, any 2-qubit unitary can be synthesized with $\le$ 3 CNOTs. 2QAN exploits this by merging 2-local Pauli strings and SWAPs into unitaries and then applying synthesis, an optimization orthogonal to \sys that could be incorporated into its circuit generator.
The frequent timeouts of 2QAN highlight its scalability limits, stemming from $O(m^2n)$ routing and $O(m^4)$ scheduling complexity (with $m$ Pauli strings and $n$ qubits)~\cite{2qan}. By contrast, \sys scales better: its ILP-based partitioner and orchestrator allow search timeouts to produce good (not exhaustive) solutions, while heuristics such as METIS offer fast approximations. Its greedy scheduler runs in $O(N_c m n)$, where $N_c$ is the number of chiplets, demonstrating substantially improved scalability over 2QAN.

The comparison with MECH~\cite{highway} is in \Cref{tab:mech}. MECH is evaluated separately because it requires both the row and column counts in a heavy-hex layout (\Cref{fig:chiplet_internal}) to be multiples of four. Thus, we use a slightly different backend topology from Backend\_1–3: each chiplet has a heavy-hex layout with 16 qubits per row and 16 rows total. The chiplet arrays for the 324- and 440-qubit circuits are $1\times3$ and $2\times2$, respectively.

As shown in \Cref{tab:mech}, \sys outperforms MECH in both circuit depth and effective CNOT count. MECH’s depth is 1.69× higher and its effective CNOT count 1.08× higher than \sys-METIS; compared to \sys-CPLEX, the gaps are 1.69× and 1.12×, respectively. \sys’s advantage stems from additional optimizations—Pauli graph partitioning and Pauli set orchestration—that improve qubit mapping by localizing interactions and shortening cross-chiplet communication. Highway reuse further boosts performance, especially for QAOA circuits.

\Cref{tab:main_results} compares \sys-CPLEX and \sys-METIS. For small problem sizes (few qubits and Pauli strings), \sys-CPLEX outperforms \sys-METIS, as seen in the three 1D Ising models, since solvers can find optimal solutions more effectively on small instances. However, this advantage fades at larger scales due to the limited scalability of the ILP approach.

\Cref{fig:time_ratio} reports compilation times of all tools across completed benchmarks. On average, Qiskit is 5.06× faster than \sys-METIS, while \tket is 2.92× faster. As discussed earlier, 2QAN is substantially slower, averaging 19.54× slower than \sys-METIS, underscoring \sys-METIS’s superior scalability over both domain-specific and general-purpose compilers. \sys-CPLEX, however, is 3.16× slower than \sys-METIS due to extensive preprocessing in the Pauli graph partitioner, despite a 10-minute solver limit. This preprocessing, which performs essential ILP reductions, is unavoidable and dominates runtime as benchmarks scale. \Cref{fig:time_cplex} details component-wise times in \sys-CPLEX, showing that preprocessing increasingly outweighs other stages. In contrast, \Cref{fig:time_metis} shows that replacing solver-based partitioning with the METIS algorithm drastically reduces partitioning time. Both figures also confirm that the set orchestrator, another solver-based step, consumes minimal time compared to other components.

\subsection{Ablation Study and Sensitivity Analysis}

\paragraph{Ablation study on \sys's Pauli graph partitioner and set orchestrator}

We compare qubit-to-chiplet assignments from \sys-CPLEX and \sys-METIS against Qiskit’s SABRE mapper~\cite{sabre} in terms of cross-chiplet communication amount and distance. As shown in \Cref{fig:comm_amount,fig:comm_dist}, Qiskit yields 2.54× more cross-chiplet communications and 2.66× greater total distance than \sys-METIS, and also exceeds \sys-CPLEX. This demonstrates that \sys’s partitioner and orchestrator produce more optimized assignments. The higher communication in \sys-CPLEX relative to \sys-METIS stems from the limited scalability of the solver-based approach.

\paragraph{Sensitivity to the number of chiplets}

Four devices with varying chiplet counts (3, 6, 12, 16) but identical chiplet sizes and intra-chiplet layouts were used to access \sys-METIS's sensitivity to the number of chiplets. The chiplet arrays for these configurations are \(1 \times 3\), \(2 \times 3\), \(3 \times 4\), and \(4 \times 4\), respectively. Circuit depth and effective CNOT count served as the evaluation metrics. The results, normalized against Qiskit's performance on the same benchmarks, are presented in \Cref{fig:scalability_depth} and \Cref{fig:scalability_eff_cx}. As depicted in these figures, \sys-METIS shows significant improvements across various benchmarks, with the performance gap between \sys-METIS and Qiskit widening as the number of chiplets increases. This indicates that \sys-METIS scales effectively with the number of chiplets.

\section{Conclusion}

This paper identifies new optimization opportunities for compiling 2-local qubit Hamiltonian simulation on chiplet-based architectures.
We introduce the concept of grouping 2-local Pauli strings into \textit{Pauli sets} and \textit{Pauli blocks}, which enable efficient implementation of multi-target controlled gates using the \textit{highway} mechanism.
Building on these insights, we present \sys, the first quantum compiler tailored for 2-local Hamiltonian simulation on quantum chiplet architectures. \sys employs a hierarchical approach to optimizing cross-chiplet communication, reducing both circuit depth and operation count. Evaluation across various tasks shows that \sys outperforms existing general-purpose and domain-specific compilers in scalability and efficiency.

\bibliographystyle{ACM-Reference-Format}
\bibliography{references}

\end{document}